\documentclass{article}

\usepackage{amssymb,amsmath,amsthm,proof}
\usepackage{graphicx}
\usepackage{gensymb}

\usepackage{multicol}
\usepackage{multirow}

\newtheorem{prop}{Proposition}[section]
\newtheorem{lem}{Lemma}[section]

\usepackage[section]{algorithm}
\usepackage{algpseudocode,algorithmicx}

\newcommand*\Let[2]{\State #1 $\gets$ #2}
\def\sds{\strut \displaystyle}

\title{On the Inversion Modulo a Power of an Integer}
\date{}
\author
{Guangwu Xu\thanks{SCST, Shandong University, China,  e-mail: {\tt gxu4sdq@sdu.edu.cn}(Corresponding author)},
	Yunxiao Tian\thanks{SCST, Shandong University, China,  e-mail: {\tt tyunxiao@mail.sdu.edu.cn}.}, Bingxin Yang\thanks{SCST, Shandong University, China, e-mail: {\tt bingxiny@mail.sdu.edu.cn}.}.\\
}

\begin{document}


\baselineskip18pt


\maketitle

\begin{abstract}
Recently, Ko\c{c} proposed a neat and efficient algorithm for computing 
\[
x = a^{-1} \pmod {p^k}
\]
for a prime $p$ based on the
exact solution of linear equations using $p$-adic expansions. The algorithm 
requires only addition and right shift per step.
In the first part of this paper, we design an algorithm that computes 
\[
x = a^{-1} \pmod {n^k}
\]
for any integers $a, n>1$ with $\gcd(a, n)=1$. The algorithm has a motivation from the schoolbook multiplication and achieves both efficiency and generality.
The greater flexibility of our algorithm is explored by  utilizing the built-in arithmetic of computer architecture,
e.g., $n=2^{64}$, and experimental results show significant improvements. 
This paper also contains some results on modular inverse based on an alternative proof of correctness of Ko\c{c} algorithm.  
For the computation of modular inverses when the modulus is a special power of a prime
$p$ (i.e., of the form $p^{2^s}$),
an efficient algorithm was developed by Dumas and later improved by Hurchalla. These methods are
based on Hensel lifting and perform particularly well when $p=2$ and $2^s$ 
matches the native bit width of a computer. In the second part of the paper, we present a generalization of these methods to moduli of the form $n^{2^s}$ for any integer $n>1$.
The derivation of our algorithm follows from a simple algebraic manipulation.


\end{abstract}

\section{Introduction}
Arithmetic serves as a fundamental building block in modern computation, supporting a variety of 
processes ranging from basic data manipulation to complex algorithmic operations. 
The efficiency of arithmetic operations, therefore, directly influences the performance and scalability of computational systems, making it a critical factor in designing high-speed and resource-efficient 
technologies.
The impact of arithmetic efficiency is particularly pronounced in specialized fields such as cryptography and communication. In cryptography, secure protocols like encryption, decryption, and digital signatures rely heavily on intensive arithmetic computations, especially those involving large integers and finite fields. 
For instance, public-key cryptosystems such as RSA and elliptic curve cryptography depend on rapid arithmetic to ensure both security and real-time responsiveness, with modular arithmetic being a major bottleneck. Recent exciting developments in post-quantum cryptographic systems and homomorphic encryption schemes also require a large number of modular operations and fast number-theoretic transforms.
The modular multiplication has emerged as a primary concern when evaluating computational costs in modern cryptographic practice.  Consequently, 
optimizing this operation through advanced algorithms or hardware acceleration is a key focus in reducing computational overhead, minimizing power usage, and enhancing the feasibility of secure communications. 

A specific modular operation $a \pmod n$ is to perform the Euclidean division
\[
a = qn +r,   \quad\quad 0\le r < n,
\]
to obtain $r$. But usual division is expensive compared to multiplication. Although there were methods to replace division with multiplication to speed up calculations even in ancient times, they were only applicable to special divisors, such as $n=2^s5^t$. In 1274, Hui Yang introduced such a method in his book \cite{Yanghui274}\footnote{There was a brief description about this book in \cite{Libbrect}.  }: without loss of generality, we assume that $s>t$, then
\[
\frac{a}n = \frac{a\times (5^{s-t})}{10^s}.
\]
Note that in decimal system, division by a power of $10$ yields the quotient and remainder simply by shifting the digits, so the expensive division is essentially replaced by the multiplication $a\times (5^{s-t})$.
However, this strategy can not be generalized to an arbitrary divisor.

The emergence of public-key cryptography introduced algorithms that rely heavily on modular arithmetic. In 1985, the famous Montgomery reduction algorithm  \cite{Mont} was proposed. This algorithm makes a very clever use of the idea of replacing division with multiplication, leveraging the fact that integer division (and the modulo operation) by a power of two is computationally straightforward. Given an odd modulus $N$ and $R=2^k>N$,
the algorithm works with  Montgomery representation for integers: $a$ is represented as $aR \pmod N$. Let $R^{-1} $ denote the modular inverse of $R$ modulo $N$ (i.e., $RR^{-1}\equiv 1 \pmod N$ ) and 
$N^{-1} $ denote the modular inverse of $N$ modulo $R$ (i.e., $NN^{-1}\equiv 1 \pmod R$ ), where
$0<R^{-1}<N$ and $0<N^{-1}<R$. First, the value $N^{-1}$ needs to be pre-computed. Then for a given input $0<T<NR$, the algorithm computes  $t$ as follows
\[
t=\frac{T+\big(\big((T\bmod R)(R-N^{-1})\big)\bmod R\big)N}{R},
\]
and finally returns
\[
TR^{-1}\bmod N = \left\{ 
\begin{array}{ll} t-N & \mbox{ if } t\ge N \\
	t& \mbox{ otherwise }.
\end{array}
\right.
\] 
For a technical motivation of Montgomery algorithm, the literature regards it as a generalization of
Hensel's odd division for 2-adic numbers \cite{BM,SV}. Alternatively, a unified treatment of Montgomery-type modular reduction algorithms using the Chinese Remainder Theorem (CRT) is presented in \cite{XJY}. To derive the integer $t$ mentioned above, we begin with the following identity involving $R$ and $N$:
\[
R^{-1}R+N^{-1}N = 1+RN.
\]
Multiplying both sides by $T$ yields
\[
TR^{-1}R+TN^{-1}N = T+TRN.
\]
This equation can be interpreted as an equivalent form of the Chinese Remainder Theorem for the two moduli $R$ and $N$. Rearranging the terms, we have $TR^{-1}R = T + T(R - N^{-1})N$, and consequently,
\[
TR^{-1}=\frac{T+T(R-N^{-1})N}{R}\equiv t \bmod N.
\]
This explains the derivation of Montgomery algorithm.

Multiplication, division, and modular arithmetic with powers of two are trivial for modern computers. What about computing the inverse of an odd integer modulo a power of two? This question is interesting not only for theoretical curiosity, but also for its practical importance. We have seen that computing inversion of an odd integer modulo a power of $2$ is a component for Montgomery modular reduction algorithm, this computational task is often seen in other scenarios, e.g. \cite{XZDDF}.

It is expected that the calculation of an inverse modulo a power of $2$ is easier than the general case. Indeed  there have been several surprisingly simple and efficient methods designed for this purpose.  
In \cite{DK},  Duss\'e and Kaliski 
gave an efficient algorithm for computing the inverse 
\[
x=a^{-1}\pmod{2^k}  
\]
by modifying
the extended euclidean algorithm. In this method, multiplication are used in each iteration.

Arazi and Qi present a review of previously existing algorithms for computing
\[
x=a^{-1}\pmod{2^k} 
\]
and  introduce a
new algorithm in their paper \cite{AQ}. For an odd integer $a$, their new contribution is an efficient algorithm for calculating 
\[
x=a^{-1}\pmod{2^{2^s}} ,
\]
by restricting the power $k$ to a special form of $2^s$. Integer multiplications are also involved in the algorithm.

In \cite{Dum}, Dumas observes that   Arazi and Qi's algorithm is a specific case of Hensel lifting. In this case, the $p$-adic analogue of the classical Newton's method for approximating real roots of a differentiable function arises naturally \cite{BS,RMurty}. Then Newton-Raphson's iteration yields a recursion to compute
\[
x=a^{-1}\pmod{p^{2^s}} ,
\]
for any prime $p$. 

An improvement of Dumas's method is presented in \cite{Hur} where  Hurchalla
describes an algorithm for the integer multiplicative inverse modulo a particular power of $2$.
The algorithm is claimed to complete in the fewest cycles known for modern microprocessors,
when using the native bit width $w$ for the modulus $2^w$.

Recently in  \cite{Koc},  Ko\c{c} proposed a neat and efficient algorithm for computing 
\[
x = a^{-1} \pmod {p^k}
\]
for any prime number $p$ and any integer $k>0$. This algorithm is based on 
Dixon's method \cite{Dixon}  for computing the exact rational solution
to a regular system of linear equations
\[
\mathbf{Ax}=\mathbf{b}
\] 
with integer coefficients. The essence of the method is to compute  integer vector $\mathbf{\tilde{x}}$
such that
\[
\mathbf{A\tilde{x}}\equiv \mathbf{b} \pmod{p^k}
\] 
for a suitably large integer $k$, then deducing the
rational solution $\mathbf{x}$ from the $p$-adic approximation $\mathbf{\tilde{x}}$.
Nice features of Ko\c{c} algorithm include that only addition (subtraction)\footnote{The algorithm does require multiplications of a digit (a non-negative integer that is smaller than $p$ ) with a number.} and right shifting (dividing by $p$) are used in each step. 
The binary version of Ko\c{c}  algorithm is significantly more efficient than the previous existing algorithms
for computing $a^{-1} \pmod {2^k}$ when k is small.
We would also like to mention that the more recent algorithm of 
Hurchalla \cite{Hur} gives a better performance when within the native word size of an architecture 
(e.g., the modulus is $2^{64}$).

One of the main purposes of this paper is to design an efficient and flexible algorithm  that 
computes 
\begin{equation}\label{eq:general}
	x = a^{-1} \pmod {n^k}
\end{equation}
for any integers $n>1$ and $k>0$. 

The motivation of our algorithm is the schoolbook multiplication. The formulation of carries in 
computation makes the use of right shifting (dividing by $n$) appear natural. The algorithm is flexible in that $n>1$ can be arbitrary, 
as long as the number $a$ to be inverted satisfies $\gcd(a,n)=1$. 
A natural application is to set $n=2^{64}$, so one can compute modulo inverse of an odd integer of the form
\[
a = a_0+a_12^{64}+a_2(2^{64})^2+\cdots+a_{k-1}(2^{64})^{k-1}, 
\]
where $0\le a_i<2^{64}$, to take full advantage of the built-in arithmetic
of a $64$-bit computer. Our 
experiments show that the resulting speed improvement over existing methods is very significant.

It is interesting to note that on a $64$-bit computer, the choice of $n=2^{128}$ is even superior to that of  
$n=2^{64}$, for most cases. Thus if the modulus is $(2^{128})^k$, one can compute modulo inverse of an 
odd integer of the form
\[
c = c_0+c_12^{128}+c_2(2^{128})^2+\cdots+c_{k-1}(2^{128})^{k-1}, 
\]
with $0\le c_i<2^{128}$, for a greater efficiency. 

There are several commonly used positional numeral systems whose radix, $n$, is not a power of a prime number (e.g., the decimal, dozenal, and sexagesimal systems). Therefore, the proposed algorithm is expected to perform better in computing (\ref{eq:general}), in line with the digital arithmetic of the system.

This paper also discusses and analyzes Hensel lifting methods of Dumas and  Hurchalla. As mentioned before, these methods are limited to moduli of the form  $p^{2^s}$ , where $p$ is prime. In this work, we generalize Hurchalla's algorithm to compute
\[
a^{-1} \pmod{n^{2^s}}
\]
for any integers $a, n>1$  with $\gcd(a, n)=1$. The derivation of our algorithm follows from a simple algebraic manipulation.

The contents of the rest of this paper are arranged as follows:
Section 2 describes Ko\c{c}'s algorithm for inversion modulo $p^k$, 
providing an alternative proof of its correctness and further discussion to highlight several useful properties of the method. 
In Section 3, the method for computing inverses modulo a general power is presented, along with its elementary motivation and formulation, as well as a discussion of its performance.
Section 4 discusses the existing methods based on Hensel lifting with a generalization.
Section 5 summarizes the paper .

\section{Ko\c{c}'s Algorithm for Inversion Mod $p^k$}\label{sec:2}
In \cite{Koc}, Ko\c{c} proposed a method for computing inverse with a modulus being a prime power.  This method is based on a procedure of $p$-adic approximation of 
the exact rational solution of a non-degenerate system of linear equations with integral coefficients due to Dixon \cite{Dixon}.

To describe Ko\c{c} algorithm, let us fix a prime $p$ and an 
integer $k>0$.  Given an integer $a>0$ such that $p\not| a$, then $x=a^{-1}\pmod {p^k}$ exists. Since $0< x <p^k$, it can be uniquely written as a $p$-adic form
\[
x= X_0+X_1p+X_2p^2+\cdots + X_{k-1}p^{k-1}
\]
with digits $0\le X_i<p$.  Ko\c{c} algorithm executes a loop of length $k$ and produces digits $X_0, X_1, \cdots, X_{k-1}$ one by one. An initial computation of 
$c = a^{-1} \pmod{p}$ is needed. Note that the number $c<p$ is of the size as a digit. 

\begin{algorithm}
	\caption{Ko\c{c} Algorithm
		\label{alg:Koc}}
	\begin{algorithmic}[1]
		\Require{A prime $p$, and an integer $k>0$; integer $0<a<p^k$ with $\gcd(p,a)=1$}
		\Ensure{$a^{-1} \pmod{p^k}$}
		\Function{ModInverse}{$a, p^k$}
		\Let{$c$}{$a^{-1}\pmod p$}
		\Let{$b_0$}{$1$}
		\For{$i = 0$ to $k-1$}
		\Let{$X_i$}{$cb_i \pmod p$}
		\Let{$b_{i+1}$}{$\frac{b_i-a  X_i}p $}
		\EndFor
		\State \Return{$x=(X_{k-1}\cdots X_1X_0)_p$}
		\EndFunction
	\end{algorithmic}
\end{algorithm}

In each step of the algorithm, essentially one subtraction is used. There are also two multiplications, $b_i c$ and $a X_i$, but the multipliers $c$ and $X_i$ are numbers less than $p$; namely, they are digits in the $p$-adic representation. The right shift $\frac{b_i - a X_i}{p}$ turns the operation of obtaining $X_i$ into finding the least significant digit of $c b_i$.

We shall provide a proof of the correctness of  algorithm  \ref{alg:Koc} 
which is  entirely different from that in \cite{Koc}. This proof will serve the purpose of
revealing some interesting properties of the computation. 
\begin{prop} Algorithm  \ref{alg:Koc} is correct.
\end{prop}
\begin{proof}
	We first note that the least significant digit of $ax$ is $1$ and digits of indices $1$ to $k-1$ are all zero. So the line $6$ becomes $pb_{i+1} = b_i-aX_i$. Therefore we have that the 
	set of relations
	\[	
	\left\{\begin{array}{lll}b_0  &=& 1\\
		pb_1 &=& b_0-aX_0\\
		pb_2  &=& b_1-aX_1\\
		pb_3  &=& b_2-aX_2\\
		& &\cdots \\
		pb_{k-1}  &=& b_{k-2}-aX_{k-2}\\
		pb_k  &=& b_{k-1}-aX_{k-1}
	\end{array}
	\right.
	\]
	implies
	\begin{equation}\label{eq:impl}
		\left\{	 \begin{array}{lll}b_0  &=& 1\\
			pb_1 &=& b_0-aX_0\\
			p^2b_2  &=& pb_1-aX_1p\\
			p^3b_3  &=& p^2b_2-aX_2p^2\\
			& &\cdots \\
			p^{k-1}b_{k-1}  &=& p^{k-2}b_{k-2}-aX_{k-2}p^{k-2}\\
			p^kb_k  &=& p^{k-1}b_{k-1}-aX_{k-1}p^{k-1}
		\end{array}
		\right.
	\end{equation}
	by multiplying a suitable power of $p$ to the both sides of each relation in the set.

	Summing up the equalities on the lower portion, one gets 
	\begin{eqnarray*}
		p^k b_k &=& 1-aX_0-aX_1p-\cdots \\
		&& -aX_{k-2}p^{k-2}-aX_{k-1}p^{k-1}\\
		&=&1-ax,
	\end{eqnarray*}
	namely, 
	\[
	ax = 1+(-b_k)p^k.
	\]
\end{proof}
It is remarked that more useful information can be revealed from this proof. To be more specific,
we have the following properties:
\begin{enumerate}
	\item [] {\bf Property 1}. For each  $s=1,2,\cdots, k-1$,  
	\begin{equation}\label{eq:ainvmodps}
		(X_{s-1}\cdots X_1X_0)_p = a^{-1}\pmod {p^s}.
	\end{equation}
	This can be checked by adding the first $s$ equalities on the right of (\ref{eq:impl}) in the proof, we get
	\begin{eqnarray*}
		p^s b_s &=& 1-aX_0-aX_1p-\cdots \\
		&&-aX_{s-2}p^{s-2}-aX_{s-1}p^{s-1}\\
		&=&1-a(X_0+X_1p+\cdots +X_{s-1}p^{s-1}),
	\end{eqnarray*}
	so $(X_{s-1}\cdots X_1X_0)_p=a^{-1}\pmod {p^s}$.
	
	It is noted that the Property 1 has been
	proved in  \cite{Koc} as the suffix property.
	Ko\c{c} also illustrates this property by the example of computing $12^{-1} \pmod {5^5}$.
	\item [] {\bf Property 2}. For each  $s=1,2,\cdots, k$,
	\begin{equation}\label{eq:psinvmoda}
		(p^s)^{-1} \pmod a =a+b_s.
	\end{equation}	
	In fact, denoting $u=(X_{s-1}\cdots X_1X_0)_p$, we have 
	\[
	p^sb_s+au=1 \mbox{ and } |b_s|<a.
	\]
	Since
	\[
	(a+b_s)p^s = 1+(p^s-u)a.
	\]
	and $b_s<0$, we see that $(p^s)^{-1} \pmod a =a+b_s.$

	This property indicates that $(p^s)^{-1} \pmod a$ is already computed and stored in variable $b_i$ as a byproduct of the 	Ko\c{c} algorithm.
	
	We now use the example from \cite{Koc}, computing 
	\[
	12^{-1} \pmod {5^5},
	\]
	to illustrate Property 2.
	
\begin{table}[h!]
		\caption{Ko\c{c}'s algorithm for Computing $12^{-1} \pmod {5^5}$}\label{tb:prop2}
		\centering
		\begin{tabular}{lll}
			$i$	& $b_{i}=\frac{b_{i-1}-a  X_i}p $ & $X_i=cb_i \pmod p$ \\
			\hline
			$0$&$b_0=1$& $X_0=3$\\
			$1$&$b_1=-7$& $X_1=4$\\
			$2$&$b_2=-11$& $X_2=2$\\
			$3$&$b_3=-7$& $X_3=4$\\
			$4$&$b_4=-11$& $X_4=2$\\
			\hline
		\end{tabular}
\end{table}
From table \ref{tb:prop2}, with the setting of $a = 12, p=5, k=5$, we get $(p^s)^{-1} \pmod a$ for $s= 1, 2, 3, 4$:
	\begin{center}
		\begin{tabular}{l}
			$p^{-1} \pmod a =5^{-1}\pmod {12} =12-7=5$\\
			$(p^2)^{-1} \pmod a=25^{-1}\pmod {12} =12-11=1$\\
			$(p^3)^{-1} \pmod a=125^{-1}\pmod {12} =12-7=5$\\
			$(p^4)^{-1} \pmod a=625^{-1}\pmod {12} =12-11=1$.
		\end{tabular}
	\end{center}
\end{enumerate}

\

For the case of  $p=2$, Ko\c{c} \cite{Koc} gets the following procedure for $a^{-1}\pmod {2^k}$ 
which is significantly more efficient than the existing algorithms (before the time of \cite{Hur})
for small $k$ :
\begin{algorithm}
	\caption{Binary Version of Ko\c{c} Algorithm
		\label{alg:Koc-2}}
	\begin{algorithmic}[1]
		\Require{An integer $k>0$ and an odd integer $a<2^k$}
		\Ensure{$a^{-1} \pmod{2^k}$}
	%
		\Function{ModInverse2}{$a, 2^k$}
		\Let{$b_0$}{$1$}
		\For{$i = 0$ to $k-1$}
		\Let{$X_{i}$}{$b_{i} \pmod 2 $}
		\Let{$b_{i+1}$}{$\frac{b_{i}-aX_{i}}2$}
		\EndFor
		\State \Return{$x=(X_{k-1}\cdots X_1X_0)_2$}
		\EndFunction
	\end{algorithmic}
\end{algorithm}

\section{An Algorithm of Inversion Mod $n^k$  for General $n$}
The main objective of this section is to develop a method for computing the modular inverse where the modulus is a power of an arbitrary integer greater than  $1$.

Let integers $n>1$ and $k>0$ be fixed. For any $0 < a < n^k$ with $\gcd(a,n) = 1$, the number $a$ has a unique radix-$n$ representation
\[
a = a_0+a_1n+a_2n^2+\cdots a_{k-1}n^{k-1},
\]
where $0\le a_i<n$.

Under these conditions, the inverse $x \equiv a^{-1} \pmod{n^k}$ exists. We express this inverse in its radix-$n$ form as
\[
x = X_0+X_1n+X_2n^2+\cdots X_{k-1}n^{k-1},
\]
with $0\le X_i<n$. 

Let $c = a^{-1}\pmod n$. From the representation above, it follows that $c = X_0$, and consequently
\[
a_0X_0\equiv 1 \pmod n.
\]

We know that $ax = 1 + n^k\ell$ for some integer $\ell \ge 0$. Writing the multiplication $ax$ using the standard schoolbook method, we see a result as illustrated in Table~\ref{tb:schmul}:

\begin{table*}[t!]
	\caption{Schoolbook Multiplication Table }	\label{tb:schmul}
	\centering
	\begin{tabular}{lllllllll}
		& & & $a_{k-1}$ & $a_{k-2}$ & $\cdots$ & $a_2$ & $a_1$ & $a_0$ \\
		& & & $X_{k-1}$ & $X_{k-2}$ & $\cdots$ & $X_2$ & $X_1$ & $X_0$ \\
		\hline
		& & & $a_{k-1}X_0$ & $a_{k-2}X_0$ & $\cdots$ & $a_2X_0$ & $a_1X_0$ & $a_0X_0$ \\
		& & $a_{k-1}X_1$ & $a_{k-2}X_1$  & $a_{k-3}X_1$ & $\cdots$ & $a_1X_1$ & $a_0X_1$& \\
		& $a_{k-1}X_2$ & $a_{k-2}X_2$ & $a_{k-3}X_2$ & $a_{k-4}X_2$ & $\cdots$  & $a_0X_2$& &\\
		&&&&&$\cdots$&&&\\
		&$a_3X_{k-2}$&$a_2X_{k-2}$&$a_1X_{k-2}$&$a_0X_{k-2}$&$\cdots$&&&\\
		&$a_2X_{k-1}$&$a_1X_{k-1}$&$a_0X_{k-1}$&&$\cdots$&&&\\
		&&&&&$\cdots$&&&\\
		$\cdots$&$q_{k}$&$q_{k-1}$&$q_{k-2}$&$q_{k-3}$&$\cdots$&$q_1$&$q_0$&\\
		\hline
		&$*$& $*$ &  $0$ & $0$ &$\cdots$&$0$&$0$&$1$\\
	\end{tabular}
	
	\noindent
	here $q_0, q_1, \cdots,q_{k-2}, q_{k-1}, \cdots$ are carries.
\end{table*}

Given that the first digit (the case $i=0$) of the radix-$n$ expansion of $ax$ is $1$, while all subsequent digits for $1 \le i < k$ are $0$, it follows that the carries $q_0, q_1, \dots, q_{k-1}$ must satisfy
$	q_0n = a_0X_0-1$ and
\[
q_{i}n =q_{i-1}+a_{i}X_0+a_{i-1}X_1+\cdots+a_0X_{i},
\]
for $1\le i <k$.

The above indicates that, assuming $X_0, X_1, \cdots, X_{i-1}$ and $q_{i-1}$  
have already been computed, the subsequent values $X_i$ and $q_i$ are determined by the
relations
\begin{align*}
	X_i &= - a_0^{-1}(q_{i-1}+a_{i}X_0+a_{i-1}X_1+\cdots+a_1X_{i-1}) \pmod n\\
	q_i &=\frac{q_{i-1}+a_{i}X_0+a_{i-1}X_1+\cdots+a_0X_{i}}{n}.
\end{align*}
To obtain a recurrence relation that is more amenable, we need 
to investigate an explicit expression of $q_i$.

By simple calculation, we have
\begin{align*}
	q_0 &= \frac{a_0X_0-1}n\\
	q_1 &=  \frac{(a_0+a_1n)(X_0+X_1n) -1 -a_1X_1n^2}{n^2}\\
	q_2 &=\frac{\sds\sum_{j=0}^2a_jn^j\sum_{j=0}^2X_jn^j -1 -\sum_{j=1}^2a_jX_{3-j}n^3-a_2X_2n^4}{n^3}\\
\end{align*}
In general, by denoting
\[
a^{(i)}=\sum_{j=0}^ia_jn^j  \mbox{ and } x^{(i)}=\sum_{j=0}^iX_jn^j ,
\]
we have a precise formula for carry:
\begin{equation}\label{eq:carry}
	q_{i} =\frac{a^{(i)}x^{(i)}-1}{n^{i+1}} - \sum_{\ell=0}^{i-1}(\sum_{j=\ell+1}^iX_ja_{i+\ell+1-j})n^{\ell}
\end{equation}
for any $0\le i<k$. Note that 
\[
a\equiv a^{(i)} \pmod{n^{i+1}},
\] 
the
formula (\ref{eq:carry}) is very suggestive and we can use $T_0=-1$ and
\[
T_{i+1} = \frac{ax^{(i)}-1}{n^{i+1}}
\]
to set up a recurrence relation.

\begin{lem}\label{lem:schoolbook} Let $T_0=-1$ and $T_{i+1} = \frac{ax^{(i)}-1}{n^{i+1}}$, then
	\[
	T_{i+1}=\frac{T_i+aX_i}n,
	\]
	and
	\[
	X_i = -a_0^{-1}T_{i}\pmod n,
	\]
	for $i=0,1,\cdots, k-1$.
\end{lem}
\begin{proof} From the definition of $T_{i+1}$, we see that
	\begin{eqnarray*}
		T_{i+1}&= &\frac{ax^{(i)}-1}{n^{i+1}} = \frac{ax^{(i-1)}-1 + aX_in^i}{n^{i+1}} \\
		&=&\frac{\frac{ax^{(i-1)}-1}{n^{i}} + aX_i}{n}=\frac{T_i+aX_i}n.
	\end{eqnarray*}
	
	Since $T_i+aX_i=n T_{i+1}$, and $a^{-1}\pmod n = a_0^{-1}\pmod n $, we obtain
	\[
	X_i = - a^{-1}T_i\pmod n =  - a_0^{-1}T_i\pmod n
	\]
	
\end{proof}

From lemma \ref{lem:schoolbook}, an algorithm for computing $a^{-1}\pmod {n^k}$ can be formulated as follows:
\begin{algorithm}[H]
	\caption{Inversion Modulo a Power} 	\label{alg:XTY}
	\begin{algorithmic}[1]
		\Require{Integers $n>1, k>0$, an integer $a<n^k$ with  $ \gcd(n,a)=1$}
		\Ensure{$a^{-1} \pmod{n^k}$}
		%
		\Function{ModInverse}{$a, n^k$}
		\Let{$c$}{$a^{-1}\pmod n$}
		\Let{$T_0$}{$-1$}
		\Let{$X_0$}{$c$}
		\For{$i = 1$ to $k-1$}
		\Let{$T_{i}$}{$\frac{T_{i-1}+X_{i-1}a}n$}
		\Let{$X_{i}$}{$-cT_{i} \pmod n $}   
		\EndFor
		\State \Return{$x=(X_{k-1}\cdots X_1X_0)_n$}
		\EndFunction
	\end{algorithmic}
\end{algorithm}

When $n=2$, the algorithm is simplified to algorithm \ref{alg:XTY-2}:

\begin{algorithm}[H]
	\caption{Inversion Modulo a Power of $2$
		}\label{alg:XTY-2}
	\begin{algorithmic}[1]
		\Require{An integer $k>0$ and an odd integer $a<2^k$}
		\Ensure{$a^{-1} \pmod{2^k}$}
		\Function{ModInverse2}{$a, 2^k$}
		\Let{$T_0$}{$0$}
		\Let{$X_0$}{$1$}
		\For{$i = 1$ to $k-1$}
		\Let{$T_{i}$}{$\frac{T_{i-1}+X_{i-1}a}2$}
		\Let{$X_{i}$}{$T_{i} \pmod 2 $}
		\EndFor
		\State \Return{$x=(X_{k-1}\cdots X_1X_0)_2$}
		\EndFunction
	\end{algorithmic}
\end{algorithm}

We should remark that from the derivation of Algorithm \ref{alg:XTY}, it is easy to verify that properties 1 and 2 from the previous section also hold for the general case of modulus $n^k$.

Our algorithm (Algorithm \ref{alg:XTY}) offers great flexibility, allowing $n$ to be any integer greater than $1$. It performs particularly well when $n$ is a power of two, such as $2^{64}$ or $2^{128}$. A detailed performance analysis for this case is provided at the end of this section. Here, we demonstrate the algorithm with a concrete example for $n=10$ by computing $65537^{-1} \pmod{10^6}$. In this case, $c$, the modular inverse of $7$ modulo $10$, is $3$ (i.e., $c = 7^{-1}\pmod{10} = 3$). The step-by-step execution of Algorithm \ref{alg:XTY} for this example is presented in the following table:
\begin{table}[h!]\label{tb:base10}
	\caption{algorithm for Computing $65537^{-1} \pmod{10^6}$}
	\centering
	\begin{tabular}{lll}
		$i$	& $T_{i}=\frac{T_{i-1}+a  X_{i-1}}n $ & $X_i=-cT_i \pmod n$ \\
		\hline
		$0$&$T_0=-1$& $X_0=3$\\
		$1$&$T_1= 19661$& $X_1=7$\\
		$2$&$T_2= 47842$& $X_2=4$\\
		$3$&$T_3= 30999$& $X_3=3$\\
		$4$&$T_4= 22761$& $X_4=7$\\
		$5$&$T_5= 41852$& $X_5=4$\\
		\hline
	\end{tabular}
\end{table}	

By collecting all the digits from this computation, we get 
\[
65537^{-1} \pmod{10^6} = 473473.
\]

In the following, we have some further discussions on algorithm \ref{alg:XTY} for the case that $n=2^{w}$.

The generality of our algorithm  implies more optimization possibilities. 
As particular examples, we consider the cases of $n=2^{64}$ and  $n=2^{128}$. 
We show that by utilizing the build-in arithmetic of a computer whose native bit width of a CPU's arithmetic instructions is $64$, a significant speed up is seen. The following table contains comparisons of our algorithm \ref{alg:XTY} for the setups of $n=2^{64}$ and  $n=2^{128}$ with (1) Ko\c{c} algorithm (algorithm \ref{alg:Koc-2}), and (2)
Hurchalla algorithm in \cite{Hur}, for (the binary length) $k=128, 256, 512, 1024, 2048, 3072$ and $4096$ (in our algorithm, these correspond to $k=2, 4, 8, 16, 32, 48$ and $64$ when  $n=2^{64}$, and $k=1, 2, 4, 8, 16, 24$ and $32$ when   $n=2^{128}$).

The experiments are performed on a system running Ubuntu 24.04.2, equipped with an AMD Ryzen 7 8845HS processor. Our implementation, as well as those of Hurchalla's algorithm \cite{Hur} and Ko\c{c}'s algorithm \ref{alg:Koc-2}, were compiled using gcc 14.2.0 with the optimization flags {\tt -Ofast -march=native -flto}. Table \ref{tb:compare} presents a performance comparison, where the values in the second, third, fourth and fifth columns denote the average runtime in nanoseconds ({\sl ns}).

\begin{table*}[t!]
	\caption{
		Comparison of Modular Inversion Algorithms for Computing $a^{-1}\pmod{2^k}$}
	\centering
	\begin{tabular}{|p{1.8cm}|p{1.9cm}|p{1.9cm}|p{1.9cm}|p{1.9cm}|}
		\hline 
		Modulus Length (bits) $k$  & Our Algorithm~\ref{alg:XTY} with $n=2^{64}$ & Our  Algorithm~ \ref{alg:XTY}  with $n=2^{128}$ & Hurchalla's Algorithm in \cite{Hur} &   Algorithm~\ref{alg:Koc-2} of Ko\c{c}\\
		\hline
		$128$ & $\mathbf{2.81}$ & - &$7.11$ & $177.89$ \\
		\hline
		$256$ & $\mathbf{22.27}$ & $\mathbf{17.71}$ & $164.48$ & $5546.50$ \\
		\hline
		$512$ & $\mathbf{68.58}$ & $\mathbf{54.98}$ &$442.70$ & $14469.55$ \\
		\hline
		$1024$ & $\mathbf{201.85}$ & $\mathbf{183.77}$ &$1538.34$ & $39118.61$ \\
		\hline
		$2048$ & $\mathbf{1302.52}$ & $\mathbf{864.64}$ &$6228.26$ & $97033.69$ \\
		\hline
		$3072$ & $\mathbf{4119.65}$ & $\mathbf{2672.20}$ &$14851.10$ & $225846.22$ \\
		\hline
		$4096$ & $\mathbf{7253.46}$ & $\mathbf{4989.27}$ &$26724.95$ & $389627.86$ \\
		\hline 
	\end{tabular}
	\label{tb:compare}
\end{table*}
In our implementation of Algorithm \ref{alg:XTY}, Hurchalla's algorithm is used to compute the initial value $c \gets a^{-1} \pmod n$. All basic arithmetic operations are performed using 64-bit native machine arithmetic, which contributes significantly to the efficiency gains observed.

We would like also to  mention that our experiments indicate that 
for larger modulus (which is a power of $2$), one can use suitable large radix $n$ to save more time.

\section{Discussion on Methods via Hensel Lifting}
In this section, we first describe and analyze Dumas' algorithm \cite{Dum} and its improvement, Hurchalla's algorithm \cite{Hur}. These algorithms compute the modular inverse of an integer, where the modulus is of the form $p^{2^s}$ for a prime $p$. Dumas' algorithm employs recursion derived from Hensel lifting, while Hurchalla's algorithm uses an equivalent recurrence but allows for a flexible initialization. In the second part of this section, we generalize Hurchalla's algorithm to compute
\[
a^{-1} \pmod{n^{2^s}}
\]
for any integers $a, n>1$  with $\gcd(a, n)=1$. We show that Hurchalla's recurrence holds in this general setting, and its derivation follows from straightforward algebraic manipulation. Furthermore, this section introduces another initialization formula for Hurchalla's algorithm.

Let $p$ be a prime number. Dumas' algorithm 
is based on using Hensel's lifting to get a solution of
\[
F_a(x) \equiv 0 \pmod {p^{2^s}},
\]
where $F_a(x) = \frac{1}{ax}-1$ and $p\not| a$. As the (formal) derivative 
\[
F_a'(x) =\frac{-1}{ax^2} \not\equiv 0 \pmod {p^{2^s}},
\]
Hensel's lemma is applicable (which is viewed as the $p$-adic analogue of the classical
method of Newton for approximating real solution of a differentiable function).
Instead of lifting the solution linearly (e.g., $p \rightarrow p^2 \rightarrow p^3 \rightarrow p^4 \cdots$), \cite{Dum} employs a lifting strategy that jumps directly to the next power of two. Starting from a root $r$ satisfying $F_a(r) \equiv 0 \pmod{p^{\ell}}$, i.e., $F_a(r) = p^{\ell} b$ for some integer $b$, the improved root is given by:
\[
F_a\bigg(r-\frac{b}{F_a'(r) }p^{\ell}\bigg) \equiv 0 \pmod {p^{2\ell}}.
\]
This process yields a Hensel lifting sequence $p \rightarrow p^2 \rightarrow p^4 \rightarrow p^8 \cdots$, doubling the exponent of the modulus at each step.

The corresponding  Newton-Raphson iteration $U_{n+1}=U_n-\frac{F_a(U_n)}{F_a'(U_n)}$ is simplified
to
\begin{equation}\label{eq:Dum-rec}
	U_{m+1} = U_m(2-aU_m)
\end{equation}
with 
\[
U_{m}= a^{-1} \pmod {p^{2^m}}
\]
for $m=1,2,\cdots, s$.
The number of steps required for the Newton-Raphson iteration is $s$, which is the logarithm (base 2) of the exponent of $p$. 
Note that the $m$-th iteration requires a multiplication ($aU_m$) and a squaring operation ($U_m^2$) on integers of size approximately $p^m$.
This impacts efficiency as $m$ grows large. Indeed, as mentioned in \cite{Dum}, this algorithm is efficient for small exponent values but becomes slower for larger exponents, e.g., those exceeding $700$ bits.

In 2022, Hurchalla proposed a modification of the method by Dumas with slightly increased generality and efficiency for the case of  $p=2$ \cite{Hur}. Consider the computation of integer multiplicative inverse  $x$
modulo $2^w$ of an odd integer $a$, namely
\[
ax\equiv 1 \pmod {2^w}.
\]
Here $w=2^s$ is typically set to the bit width of arithmetic operations on a computer. The method first chooses a small value $i_0$ such that $\frac{w}{i_0}$ is a power of $2$ and computes an initial approximation  $x_0$ satisfying
\[
ax_0 \equiv 1 \pmod {2^{i_0}}.
\]
Let $y$ be any integer satisfying $y\equiv 1-ax_0 \pmod {2^w}$, and 
define the recurrence
\begin{equation}\label{eq:Hur-rec}
	x_{m+1} \equiv x_m(1+y^{2^m}) \pmod {2^w}.
\end{equation}
Then, $x_{r}$ is the integer multiplicative inverse of $a \pmod {2^w}$, where $r=\log_2 \frac{w}{i_0}$.

It is noted in \cite{Hur} that $1-y^{2^m}=ax_m \pmod {2^w}$; consequently, 
\[
1+y^{2^m}=2-ax_m \pmod {2^w}.
\]
Therefore,
the recurrences (\ref{eq:Dum-rec}) and (\ref{eq:Hur-rec}) are equivalent. 

Explicit formulas for $x_0$
for the cases $i_0=1,2,4$ are given in \cite{Hur}. For the nontrivial case of $i_0=4$, the two formulas proposed in \cite{Hur} are
\[\begin{array}{l}
	(a) \	x_0 = \mbox{XOR} (3a \pmod{2^w}, 2)\\
	(b) \	x_0 = \big(\mbox{XOR} (a, 2)-(a+a)\pmod{2^w}\big)\pmod{2^w}.
\end{array}
\]

In this paper, we propose the following explicit formula for $x_0$ when $i_0=4$:
Suppose $(\cdots a_5a_4a_3a_2a_11)_2$ is the binary representation of the odd integer
$a$. Then
\begin{equation}\label{eq:xu-1}
	x_0 = 1 + a_1 2 + a_2 2^2 + \delta_3 2^3
\end{equation}
with $\delta_3 = (a_1 + a_2 + a_3) \pmod{2}$ satisfying $a x_0 \equiv 1 \pmod{16}$.

This formula  might be useful in some applications as it only requires computing the XOR of the bits $a_1$, $a_2$, and $a_3$.
We shall prove a more general explicit formula for $a^{-1}\pmod {2^5}$. Formula (\ref{eq:xu-1}) will be a consequence of this result.
Write
\[
a=1+a_1 2+a_2 2^2+a_3 2^3+a_42^4+a_52^5+\cdots ,
\]
where $a_i\in \{0,1\}$,
we verify that the integer
\begin{equation}\label{eq:xu}
	x = 1+a_1 2+a_2 2^2+\delta_3 2^3+\delta_4 2^4
\end{equation}
with $\delta_3, \delta_4$ being given by
\begin{align*}
	\delta_3 &=(a_1+a_2+a_3)\pmod 2 \\
	\delta_4 &=(a_1+a_2+a_4+\frac{a_1+a_2+a_3+\delta_3}2)\pmod 2
\end{align*}
yields 
\[
ax\equiv 1 \pmod{2^5}.
\]
Indeed, since $\delta_3=(a_1+a_2+a_3)\pmod 2$, $a_1+a_2+a_3+\delta_3$ is an even number, and
$a_1(a_2+a_3+\delta_3)\equiv a_1 \pmod 2$, we have
\[
\delta_4\equiv a_2+a_4+a_1(a_2+a_3+\delta_3)+\frac{a_1+a_2+a_3+\delta_3}2\pmod 2.
\]
Therefore, \small{
	\begin{align*}
		ax&=1+2^3(a_1+a_2+a_3+\delta_3)\\
		& +2^4(a_2+a_4+a_1(a_2+a_3+\delta_3)+\delta_4)+2^5(a_5+\cdots )\\
		&\equiv 1+2^4(\frac{a_1+a_2+a_3+\delta_3}2+a_1+a_2+a_4+\delta_4) \pmod {2^5}\\
		&\equiv 1 \pmod {2^5}.
\end{align*}}\normalsize{}
This verifies the validity of (\ref{eq:xu}). By the property 2 in section \ref{sec:2}, since $(\delta_4 \delta_3 a_2a_1 1)_2$ is the multiplicative inverse of $a=(\cdots a_5a_4a_3a_2a_11)_2$ modulo $2^5$, we see that 
\begin{center}
	\begin{tabular}{l}
		$a^{-1} \pmod {2^5} =(\delta_4 \delta_3 a_2a_1 1)_2$\\
		$a^{-1} \pmod {2^4} =(\delta_3 a_2a_1 1)_2$\\
		$a^{-1} \pmod {2^3} =(a_2a_1 1)_2$\\
		$a^{-1} \pmod {2^2} =(a_1 1)_2$.
	\end{tabular}
\end{center}

Finally in this section we generalize Hurchalla's method to compute
\begin{equation}\label{eq:xu-gen}
	a^{-1}\pmod {n^{2^s}}
\end{equation}
for any positive integers $a, n$ with $\gcd(a,n)=1$. Our approach is built on the following simple algebraic derivation.

\begin{lem}\label{lem:xl}
	Let $a, n$ be positive integers with $\gcd(a,n)=1$. Denote by $u_0$ the multiplicative inverse of $a \pmod {n}$.
	\begin{enumerate}
		\item Let $k>0$ be an integer and set $\phi=1-au_0\pmod{n^k}$, then
		\[
		T_m =  u_0(1+\phi+\phi^2+\cdots+\phi^{m-1}) \pmod{n^m}
		\]
		is the multiplicative inverse of $a \pmod {n^m}$ for any $m\le k$.
		\item
		Let $s, h_0>0$ be  integers with $0\le h_0 < s$. Set $x_0=T_{2^{h_0}}$. Denote $y=1-ax_0\pmod{n^{2^s}}$ and define the recurrence
		\begin{equation}\label{eq:rec}
			x_{m+1}=x_m(1+y^{2^m}) \pmod{n^{2^s}}.
		\end{equation}
		Then, for $r=s-h_0$,  
		\[
		ax_{r}\equiv 1  \pmod {n^{2^s}}.
		\]
	\end{enumerate} 
\end{lem}

\begin{proof}
	\begin{enumerate}
		\item  Since $au_0\equiv 1 \pmod n$, there exists an integer $v_0$ such that
		\[
		1-au_0=v_0 n.
		\]
		Note that $au_0 \equiv 1-\phi \pmod{n^k}$ and $\phi\equiv v_0 n \pmod{n^k}$. Consequently,
		\[
		au_0 \equiv 1-\phi \pmod{n^m} \ \mbox{ and } \
		\phi\equiv v_0 n \pmod{n^m}
		\]
		since $m\le k$. Therefore
		\begin{eqnarray*}
			a T_m &\equiv&au_0(1+\phi+\cdots+\phi^{m-1}) \pmod{n^m}\\
			&\equiv&(1-\phi) (1+\phi+\cdots+\phi^{m-1}) \pmod{n^m}\\
			&\equiv& 1-\phi^m \equiv 1-(v_0n)^m \pmod{n^m}\\
			&\equiv& 1 \pmod{n^m}.
		\end{eqnarray*}
		\item Since $x_0=T_{2^{h_0}}$,  part 1) implies
		\[
		ax_0\equiv 1  \pmod {n^{2^{h_0}}},
		\]
		i.e., there exists an integer $\eta$ such that
		\[
		y \equiv 1-ax_0 \equiv  \eta n^{2^{h_0}} \pmod{n^{2^s}}.
		\]
		
		From $ x_1=x_0(1+y) \pmod{n^{2^s}}$, we obtain
		\begin{eqnarray*}
			x_2 &=&x_1(1+y^2) \pmod{n^{2^s}} \\
			&=&x_0(1+y+y^2+y^3) \pmod{n^{2^s}}\\
		\end{eqnarray*}
		Inductively, we get 
		\[
		x_i\equiv x_0(1+y+y^2+\cdots+y^{2^{i}-1}) \pmod{n^{2^s}}.
		\]
		Therefore
		\begin{eqnarray*}
			a x_r &\equiv &ax_0(1+y+\cdots+y^{2^{r}-1})\pmod{n^{2^s}}\\
			&\equiv&(1-y)(1+y+\cdots+y^{2^{r}-1})\pmod{n^{2^s}}\\
			&\equiv& 1-y^{2^{r}}\equiv 1-( n^{2^{h_0}})^{2^{r}}\pmod{n^{2^s}}\\
			&\equiv& 1 \pmod{n^{2^s}}.
		\end{eqnarray*}
	\end{enumerate} 
\end{proof}
The essence of part 2) of the lemma is a divide-and-conquer strategy. By splitting the right hand side of the relation
\[
x_{m+1}\equiv x_0(1+y+y^2+\cdots+y^{2^{m+1}-1}) \pmod{n^{2^s}}
\]
into $x_0(1+y+y^2+\cdots+y^{2^{m}-1})$ and  $y^{2^m}x_0(1+y+y^2+\cdots+y^{2^{m}-1})$, we obtain the recurrence 
(\ref{eq:rec}). Naturally, part 2) of Lemma~\ref{lem:xl}
yields the following generalization of Hurchalla's algorithm
\begin{algorithm}[H]
	\caption{Generalization of Hurchalla's Algorithm} 	\label{alg:XG}
	\begin{algorithmic}[1]
		\Require{Integers $n>1, s>0$, an integer $a<n^{2^s}$ with  $ \gcd(n,a)=1$}
		\Ensure{$a^{-1} \pmod{n^{2^s}}$}
		%
		\Function{SpecialModInverse}{$a, n^{2^s}$}
		\State {fix a small integer $0\le h_0<s$}
		\Let{$x_0$}{$a^{-1}\pmod {n^{2^{h_0}}}$}
		\Let{$y$}{$1-ax_0 \pmod{n^{2^s}}$}
		\For{$i = 1$ to $s-h_0$}
		\Let{$x_{i}$}{$x_{i-1}(1+y) \pmod{n^{2^s}}$}
		\Let{$y$}{$y^2 \pmod{n^{2^s}}$}   
		\EndFor
		\State \Return{$x_{s-h_0}$}
		\EndFunction
	\end{algorithmic}
\end{algorithm}

\section{Conclusions}
In this paper, we first develop a method for computing the modular inverse modulo a power of an arbitrary integer greater than one. In structure, our algorithm is comparable to a recent algorithm by
Ko\c{c} for prime powers, but it draws inspiration from the elementary schoolbook multiplication. 
Since the base (or radix) of our modulus is unrestricted, we are able to compute $a^{-1}\pmod {n^k}$ for any integer $n>1$. Experimental results using 
$n=2^{64}$ and $n=2^{128}$ demonstrate a significant performance advantage over existing methods. The paper also provides an alternative proof of correctness of Ko\c{c}'s algorithm and explains how to use the proof to derive both $a^{-1} \pmod {n^{s}}$ and $(n^s)^{-1}\pmod a$ for $s<k$ directly from the procedure for computing $a^{-1}\pmod {n^k}$.

The second part of the paper analyzes Hensel lifting method by Dumas and its subsequent improvement by Hurchalla for computing $a^{-1} \pmod{p^{2^s}}$. We propose a simple algebraic derivation that generalizes these 
methods to moduli of the form $n^{2^s}$ for any integer $n>1$.



\begin{thebibliography}{9}
\itemsep=-2pt
\small
\bibitem{AQ} O. Arazi and H. Qi, ``On calculating multiplicative inverses modulo $2^{m}$,''\textit{IEEE Trans. Comput.}, vol.~57, no.~10, pp.~1435--1438, Oct.~2008.
\bibitem{BS} E. Bach and J. Shallit, ``Algorithmic Number Theory: Efficient Algorithms'',
Cambridge, MA: MIT Press, 1996.
\bibitem{BM} J. W. Bos and P. L. Montgomery, ``Montgomery Arithmetic from a Software Perspective'', In: Bos JW, Lenstra AK, eds. Topics in Computational Number Theory Inspired by Peter L. Montgomery. Cambridge University Press; 2017:10-39. 
\bibitem{Dixon} J. D. Dixon,``Exact solution of linear equations using p-adic expansions,''
\textit{Numerische Mathematik}, vol.~40, no.~1, pp.~137--141, 1982.
\bibitem{Dum} J. Dumas, ``On Newton-Raphson iteration for multiplicative inverses modulo prime powers,''\textit{IEEE Trans. Comput.}, vol.~63, no.~8, pp.~2106--2109, Aug.~2014.
\bibitem{DK}  S. R. Duss\'e and B. S. Kaliski Jr, ``A cryptographic library for the Motorola DSP56000,'' in \textit{Proc. Workshop Theory Appl. Cryptographic Techn.}, 1990, pp.~230--244.

\bibitem{Hur} J. Hurchalla, 
``An Improved Integer Modular Multiplicative Inverse (modulo $2^w$),'' Comput. Res. Repos. (2022). https://arxiv.org/abs/2204.0434.
\bibitem{Koc} \c{C}. K. Ko\c{c}, ``Algorithms for Inversion Mod $p^k$,''\textit{IEEE Trans. Comput.}, vol.~69, no.~6, pp.~907--913, Jun.~2020.  
\bibitem{Libbrect} U. Libbrect,
Chinese Mathematics in the Thirteenth Century, {\sl Dover Publications}, 2005.
\bibitem{Mont} P. L. Montgomery, ``Modular multiplication without trial division,'' \textit{Math. Comput.}, vol.~44, no.~170, pp.~519--521, Apr.~1985.
\bibitem{RMurty} M. R. Murty, ``Introduction
to p-adic Analytic
Number Theory'', AMS/IP Studies in Advanced Mathematics, 2009.
\bibitem{Yanghui274} H. Yang, ``Complete Mastery of Metamorphoses in 
Multiplication and Division'', 1274. (In Chinese)
\bibitem{SV} M. Shand and J. Vuillemin, ``Fast implementations of RSA cryptography'', In E. E. S. Jr., M. J. Irwin, and G. A. Jullien, editors, 11th Symposium on Computer Arithmetic, pages 252-259. IEEE Computer Society, 1993.
\bibitem{XZDDF} B. Xiang, J. Zhang, Y. Deng, Y. Dai, D. Feng,  ``Fast Blind Rotation for Bootstrapping FHEs,'' CRYPTO 2023, pp 3-36.
\bibitem{XJY} G. Xu, Y. Jia, Y. Yang, ``Chinese Remainder Theorem Approach to Montgomery-Type Algorithms'', 2025.
https://arxiv.org/html/2402.00675v3
\end{thebibliography}
\end{document}